%% file: Schmidt.tex
\documentclass[proceedings]{stacs}
\stacsheading{2010}{633-644}{Nancy, France}
\firstpageno{633}

\usepackage{graphicx}
\usepackage{wrapfig}
\usepackage{subfigure}
\usepackage{enumitem}

\newcommand{\titel}{Construction Sequences and Certifying 3-Connectedness}

\usepackage{xspace} 
\newcommand{\BG}{BG\xspace}

\usepackage{amsmath} 
\usepackage{amsfonts} 

\begin{document}
	\title{\titel}
	\author{J. M. Schmidt}{Jens M. Schmidt}
	\address{Dept. of Computer Science, Freie Universit\"at, Berlin, Germany}
	\email{jens.schmidt@inf.fu-berlin.de}
	\thanks{This research was supported by the Deutsche Forschungsgemeinschaft within the research training group ``Methods for Discrete Structures'' (GRK 1408).}
	\keywords{Algorithms and data structures, construction sequence, 3-connected, certifying algorithm, Tutte contraction, removable edges, \emph{ACM classification}: F.2.2;G.2.2}

\begin{abstract}
\input{Abstract.tex}
\end{abstract}
\maketitle



\section{Introduction}
\input{Introduction.tex}

\section{Construction Sequences}\label{constructionSequences}
\input{Preliminaries.tex}

\section{Prescribing Subdivisions}\label{prescribing}
\input{Existence.tex}

\section{Representations}\label{representations}
\input{Computational.tex}

\section{Certifying and Testing $3$-Connectedness in $O(n^2)$}\label{certifying}
\input{Certifying.tex}


\bibliographystyle{plain}
\bibliography{Schmidt}


\end{document}

%% file: Abstract.tex

Tutte proved that every $3$-connected graph on more than $4$ nodes has a \emph{contractible edge}. Barnette and Gr\"unbaum proved the existence of a \emph{removable edge} in the same setting. We show that the sequence of contractions and the sequence of removals from $G$ to the $K_4$ can be computed in $O(|V|^2)$ time by extending Barnette and Gr\"unbaum's theorem. As an application, we derive a certificate for the $3$-connectedness of graphs that can be easily computed and verified.

%% file: Introduction.tex

Instead of dealing with contractions or removals in a $3$-connected graph $G = (V,E)$ we take the equivalent view of starting with the complete graph on four vertices $K_4$ and applying their inverse operations until $G$ is constructed. Such a sequence is called a \emph{construction sequence} of $G$. We will define contractions, removals and their inverse operations in Section~\ref{constructionSequences}.

Although existence theorems on contractible and removable edges are used frequently in graph theory~\cite{Tutte1966,Thomassen1981,Thomassen2006}, we are not aware of any computational results to find the whole construction sequence, except when contractions and removals are allowed to intermix~\cite{Albroscheit2006}. Moreover, efficient algorithms are unlikely to be derived from the existence proofs as they, e.\,g., in the case of Barnette and Gr\"unbaum, depend heavily on adding longest paths, which are NP-hard to find. In contrast, we show that it is possible to find a construction sequence for a graph $G$ in time $O(|V|^2)$ for Barnette and Gr\"unbaum's characterization, at the expense of having parallel edges in intermediate graphs. In addition, we show that Barnette and Gr\"unbaum's sequence can be transformed in linear time to Tutte's sequence of contractions and is therefore algorithmically at least as powerful. Both algorithms do not rely on the $3$-connectedness test of Hopcroft and Tarjan~\cite{Hopcroft1973}, which runs in linear time but is rather involved.

Blum and Kannan~\cite{Blum1989} introduced the concept of \emph{certifying algorithms}, which give an easy-to-verify proof of correctness along with their output. While being important for program verification, certifying algorithms provide often new insights into a problem, which can lead to new methods. For that reasons they are a major goal for problems on which the fast solutions known are complicated and difficult to implement. Testing a graph on $3$-connectedness is such a problem, but surprisingly few work has been devoted to certifying algorithms, although a sophisticated linear-time algorithm without certificates is known for over $35$ years~\cite{Hopcroft1973, Vo1983, Vo1983a}. In fact, we are aware of only one certifying algorithm for that problem~\cite{Albroscheit2006}, which runs in quadratic time, but is quite involved. Using construction sequences, we give a simple, alternative solution with running time $O(|V|^2)$ and show that the used certificate is easy to verify in time $O(|E|)$.

We first recapitulate well-known results on the existence of construction sequences in Sections~\ref{tuttecharacterization} and~\ref{barnettecharacterization} and point out how Tutte's sequence can be obtained from Barnette and Gr\"unbaum's sequence in linear time. Sections~\ref{identifying} and~\ref{prescribing} cover the main idea for the existence result that we use for computing Barnette and Gr\"unbaum's sequence. Section~\ref{representations} deals with the question how construction sequences are efficiently represented and Section~\ref{certifying} shows how to use construction sequences for a certifying $3$-connectedness test.

%% file: Preliminaries.tex
Let $G=(V,E)$ be a finite graph with $n := |V|$, $m := |E|$, $V(G) = V$ and $E(G) = E$. A graph is \emph{connected} if there is a path between any two nodes and \emph{disconnected} otherwise. For $k \geq 1$, a graph is \emph{$k$-connected} if $n > k$ and deleting every $k-1$ nodes leaves a connected graph. A node (a pair of nodes) that leaves a disconnected graph upon deletion is called a \emph{cut vertex} (a \emph{separation pair}). Note that $k$-connectedness does not depend on parallel edges nor on self-loops. A path leading from node $v$ to node $w$ is denoted by $v \rightarrow w$. For a node $v$ in a graph, let $N(v) = \{w \mid vw \in E\}$ denote its set of neighbors and $deg(v)$ its degree. For a graph $G$, let $\delta(G)$ be the minimum degree of its vertices.

A \emph{subdivision} of a graph replaces each edge by a path of length at least one. Conversely, we want a notation to get back to the graph without subdivided edges. If $deg(v)=2$, $|N(v)|=2$ and $v \notin N(v)$ for a graph $G$, let $\emph{smooth}_v(G)$ be the graph obtained from $G$ by deleting $v$ followed by adding an edge between its neighbors; we say $v$ is \emph{smoothed}. If one of the conditions is violated, let $\emph{smooth}_v(G) = G$. Let $\emph{smooth}(G)$ be the graph obtained by smoothing every node in $G$. For an edge $e \in E$, let $G \setminus e$ denote the graph obtained from $G$ by deleting $e$. Let $K_n$ be the complete graph on $n$ nodes.

The following are well-known corollaries of Menger's theorem~\cite{Menger1927}.



\begin{lemma}\label{fanlemma}\emph{(Fan Lemma)}
Let $v$ be a node in a graph $G$ that is $k$-connected with $k \geq 1$ and let $A$ be a set of at least $k$ nodes in $G$ with $v \notin A$. Then there are $k$ internally node-disjoint paths $P_1,\ldots,P_k$ from $v$ to distinct nodes $a_1,\ldots,a_k \in A$ such that for each of these paths $V(P_i) \cap A = a_i$.
\end{lemma}

\begin{lemma}\label{expansionlemma}\emph{(Expansion Lemma~\cite{West2001})}
Let $G$ be a $k$-connected graph. Then the graph obtained by adding a new node $v$ joined to at least $k$ nodes in $G$ is still $k$-connected.
\end{lemma}

\subsection{Tutte's Characterization and their Inverse}\label{tuttecharacterization}
From now on we assume for simplicity that our input graph $G=(V,E)$ is simple although all results can be extended to multigraphs. Generally, contractions cannot always avoid parallel edges in intermediate graphs, e.\,g., for wheels. That is why we define contractions to preserve graphs to be simple: \emph{Contracting} an edge $e=xy$ in a graph deletes $e$, identifies nodes $x$ and $y$ and replaces iteratively all $2$-cycles by an edge. An edge $e$ is called \emph{contractible} if contracting $e$ results in a $3$-connected graph.

A \emph{node splitting} takes a node $v$ of a $3$-connected graph, replaces $v$ by two nodes $x$ and $y$ with an edge between them and replaces every former edge $uv$ that was incident to $v$ with either the edge $ux$, $uy$ or both such that $|N(x)| \geq 3$ and $|N(y)| \geq 3$ in the new graph. Node splitting as defined here is therefore the exact inverse of contracting a contractible edge that has on both endnodes at least $3$ neighbors.

\begin{theorem}\label{tutte}\emph{(Corollary of Tutte~\cite{Tutte1961})}
The following statements are equivalent:
\begin{align}
	&\ \text{A simple graph $G$ is $3$-connected}\notag\\
	\Leftrightarrow &\ \exists \text{ sequence of contractions from $G$ to $K_4$ on contractible edges $e=xy$}\notag\\
	&\ \text{with $|N(x)| \geq 3$ and $|N(y)| \geq 3$}\label{contractions}\\
	\Leftrightarrow &\ \exists \text{ construction sequence from $K_4$ to $G$ using node splittings}\label{nodesplittings}
\end{align}
\end{theorem}

We describe next a straight-forward $O(n^2)$ algorithm to compute~\eqref{contractions} for a graph $G$ on more than $4$ vertices. First, we decrease the number of edges to $O(n)$ in $G$ by applying the algorithm of Nagamochi and Ibaraki~\cite{Nagamochi1992}. This preserves the $3$-connectedness or respectively, the non $3$-connectedness of $G$. Moreover, it is known that the resulting graph contains a vertex $v$ of degree $3$. By a result of Halin~\cite{Halin1969a}, every node of degree $3$ is incident to a contractible edge $e$. We get $e$ by subsequently contracting each of the three incident edges and testing the resulting graph with the algorithm of Hopcroft and Tarjan~\cite{Hopcroft1973} for $3$-connectedness. Iteration of both subroutines gives us the whole contraction sequence in $O(n^2)$ time. However, the Hopcroft-Tarjan test is difficult to implement and we will give a much simpler algorithm that is capable of computing both characterizations later.

\subsection{Barnette and Gr\"unbaum's Characterization and their Inverse}\label{barnettecharacterization}
The Barnette and Gr\"unbaum operations (\emph{\BG-operations}) consist of the following operations on a $3$-connected graph (see Figures~\ref{fig:BG1}-\ref{fig:BG3}).

\begin{enumerate}[label=(\alph*)]
	\item add an edge $xy$ (possibly a parallel edge)
	\label{operation1}
	\item subdivide an edge $ab$ by a node $x$ and add the edge $xy$ for a node $y \notin \{a,b\}$
	\label{operation2}
	\item subdivide two distinct, non-parallel edges by nodes $x$ and $y$, respectively, and add the edge $xy$
	\label{operation3}
\end{enumerate}

In all three cases, let $xy$ be the edge that was \emph{added} by the \BG-operation.

\begin{figure}[htb]
	\centering
	\subfigure[parallel edges allowed]{
		\includegraphics[scale=0.6]{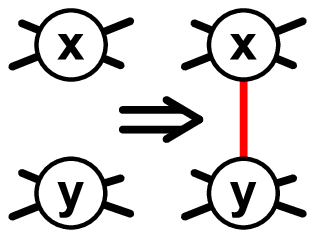}
		\label{fig:BG1}
	}
	\hfill
	\subfigure[$y$, $a$, $b$ distinct]{
		\includegraphics[scale=0.6]{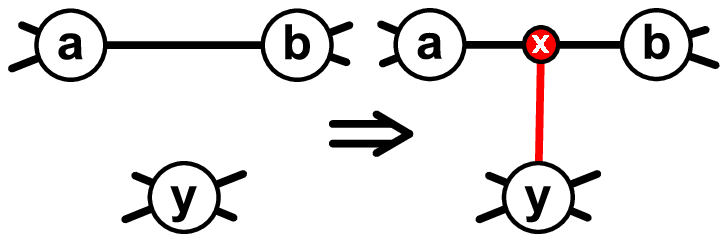}
		\label{fig:BG2}
	}
	\hfill
	\subfigure[$e \neq f$, $e$ and $f$ not parallel]{
		\includegraphics[scale=0.6]{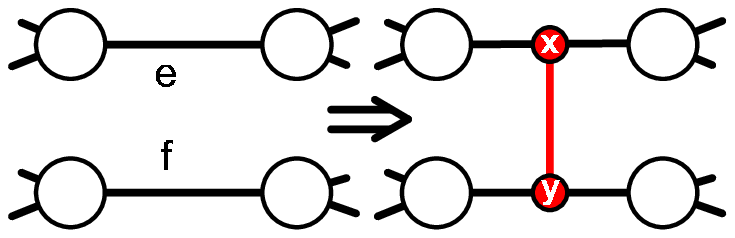}
		\label{fig:BG3}
	}
	\caption{The three operations of Barnette and Gr\"unbaum.}
	\label{fig:BGOperations}
\end{figure}

%

\begin{theorem}\label{barnettetheorem}\emph{(Barnette and Gr\"unbaum~\cite{Barnette1969}, Tutte~\cite{Tutte1966})}
A graph $G$ is $3$-connected if and only if $G$ can be constructed from the $K_4$ using \BG-operations.
\end{theorem}

Theorem~\ref{barnettetheorem} was proven in this notation by Barnette and Gr\"unbaum~\cite{Barnette1969}, but implicitly described in a theorem about \emph{nodal connectivity} by Tutte~\cite[Theorem~$12.65$]{Tutte1966}. If not stated otherwise, every construction sequence uses only \BG-operations. Let a \BG-operation be \emph{basic}, if it does not create parallel edges and let a construction sequence be \emph{basic}, if it only uses basic \BG-operations.

Like in Theorem~\ref{tutte}, we want the inverse of a \BG-operation. Let \emph{removing} the edge $e=xy$ of a graph be the operation of deleting $e$ followed by smoothing $x$ and $y$. An edge $e=xy$ in $G$ is called \emph{removable}, if removing $e$ yields a $3$-connected graph. We show that removing a removable edge $e=xy$ with $|N(x)| \geq 3$, $|N(y)| \geq 3$ and $|N(x) \cup N(y)| \geq 5$ is exactly the inverse of a \BG-operation.

\begin{theorem}\label{removing}
The following statements are equivalent:
\begin{align}
	&\ \text{A simple graph $G$ is $3$-connected}\label{triconnected}\\
	\Leftrightarrow &\ \exists \text{ sequence of removals from $G$ to $K_4$ on removable edges $e=xy$}\notag\\
	&\ \text{with $|N(x)| \geq 3$, $|N(y)| \geq 3$ and $|N(x) \cup N(y)| \geq 5$}\label{removals}\\
	\Leftrightarrow &\ \exists \text{ construction sequence from $K_4$ to $G$ using \BG-operations}\label{BGConstruction}\\
	\Leftrightarrow &\ \exists \text{ basic construction sequence from $K_4$ to $G$ using \BG-operations}\label{BGBasic}
\end{align}
\end{theorem}

\begin{proof}

Theorem~\ref{barnettetheorem} establishes $\eqref{triconnected} \Leftrightarrow \eqref{BGConstruction}$. Moreover, the proof of Theorem~\ref{barnettetheorem} in~\cite{Barnette1969} implicitly shows that on simple graphs basic operations suffice, thus only the equivalence for~\eqref{removals} remains. We first prove $\eqref{BGBasic} \Rightarrow \eqref{removals}$ and then $\eqref{removals} \Rightarrow \eqref{BGConstruction}$.

\BG-operations operate by definition on $3$-connected graphs, this holds in particular for the ones in~\eqref{BGConstruction}. Let $G'$ be the graph obtained by a basic \BG-operation in~\eqref{BGConstruction} that adds the edge $e=xy$. The operation can clearly be undone by removing $e$ in $G'$. Since \BG-operations preserve $3$-connectedness with Theorem~\ref{barnettetheorem}, $|N(x)| \geq 3$ and $|N(y)| \geq 3$ hold in $G'$.

It remains to show that $|N(x) \cup N(y)| \geq 5$ in $G'$. If $|N(x)| \geq 4$ or $|N(y)| \geq 4$, $|N(x) \cup N(y)| \geq 5$ follows, since $x$ and $y$ are neighbors and no self-loops exist. Thus, let $|N(x)| = |N(y)| = 3$. Having $N(x) \setminus \{y\} \neq N(y) \setminus \{x\}$ yields $|N(x) \cup N(y)| \geq 5$ as well, so let $N(x) \setminus \{y\}$ and $N(y) \setminus \{x\}$ contain the same two nodes $a$ and $b$. If $|V(G)| > 4$, $a$ or $b$ must be adjacent to a node $c$ that is neither adjacent to $x$ nor $y$. But then $\{a,b\}$ is a separation pair, contradicting the $3$-connectedness of $G$. On the other hand, $|V(G)|=4$ is not possible, since that implies the \BG-operation to be~\ref{operation1} (since only~\ref{operation2} and~\ref{operation3} create new vertices) and that is no basic operation on the $K_4$.

We prove $\eqref{removals} \Rightarrow \eqref{BGConstruction}$. Let $G'$ be the graph containing a removable edge $e=xy$ that is removed in~\eqref{removals}. Note that $G'$ can have parallel edges due to previous removals but no self-loops. The removal can be undone by one of the \BG-operations. Which one, is dependent on the number $i$ of endnodes of $e$ on which smoothing changed the graph, i.\,e., the number of endnodes $u$ of $e$ with $|N(u)|=deg(u)=3$ in $G'$. If $i=0$, removing $e$ just deletes $e$ which is inversed by operation~\ref{operation1}. For $i=1$, let $x$ be the node with $|N(x)|=deg(x)=3$ in $G'$ and $f$ be the edge in which $x$ was smoothed. Then~\ref{operation2} can be applied, because $y \notin f$ (see Figure~\ref{fig:check1}) since otherwise $x$ would have had only $2$ neighbors in $G'$, contradicting the assumption $|N(x)| \geq 3$.

If $i=2$, let $f_1$ and $f_2$ be the edges in which $x$ and $y$ were smoothed. Operation~\ref{operation3} can only be applied if $f_1$ and $f_2$ are neither identical (see Figure~\ref{fig:check2}) nor parallel. But $f_1 = f_2$ would again contradict $|N(x)| \geq 3$ in $G'$ and $f_1$ being parallel to $f_2$ would contradicts $|N(x) \cup N(y)| \geq 5$ in $G$, since in that case $x$ and $y$ are only adjacent to each other and the two nodes $f_1 \cap f_2$.
\end{proof}

We show that Barnette and Gr\"unbaum's characterization is algorithmically at least as powerful as Tutte's by giving a simple linear time transformation. Lemma~\ref{transformation} allows us to focus on computing \BG-operations only.

\begin{lemma}\label{transformation}
Every construction sequence using \BG-operations can be transformed in linear time to Tutte's sequence~\eqref{contractions} of contractions.
\end{lemma}
\begin{proof}
We transform every \BG-operation in reverse order of the construction sequence to $0$, $1$ or $2$ contractions each. Operation~\ref{operation1} yields no contraction while operation~\ref{operation2} yields the contraction of exactly one part of the subdivided edge (either $xa$ or $xb$ in Figure~\ref{fig:BGOperations}). For an operation~\ref{operation3}, let $e=ab$ and $f=vw$ be the edges that are subdivided with $x$ and $y$. Both edges share at most one node; let w.\,l.\,o.\,g. $a=v$ be that node if it exists. We create one contraction for each of the edges $xb$ and $yw$ in arbitrary order. In all cases, contractions inverse \BG-operations except for the added edge $xy$, which is left over. But additional edges do not harm the $3$-connectedness of the graph nor subsequent contractions. Thus, we have found a contraction sequence to the $K_4$ unless the first contraction in the case of an operation~\ref{operation3} yields at some point a graph $H$ that is not $3$-connected. But $H$ can be obtained from the graph that results from contracting the second edge by applying one operation~\ref{operation2} and therefore is $3$-connected.
\end{proof}

\subsection{Identifying Intermediate Graphs with Subdivisions in $G$}\label{identifying}
Let $K_4=G_0,G_1,\ldots,G_z=G$ be the $3$-connected graphs obtained in a construction sequence $Q$ to a simple $3$-connected graph $G$ using the basic \BG-operations $C_0,\ldots,C_{z-1}$. We can reverse $Q$ by starting with $G$ and removing the added edges of \BG-operations in reverse order. Suppose we would delete the added edge of every $C_i$ instead of removing it and treat emerging paths containing interior nodes of degree $2$ as (topological) edges in $G_i$ (see Figure~\ref{fig:graphsequences}). Then iteratively paths are deleted instead of edges being removed and we obtain the sequence of subdivisions $G=S_z,\ldots,S_0$ in $G$ with $S_0$ being a subdivision of the $K_4$. This leads to the following observation.

\begin{lemma}[Observation]\label{observation}
Let $Q$ be a construction sequence from a graph $G_0$ to $G$ using \BG-operations. Then $G$ contains a subdivision of $G_0$ that is specified by $Q$.
\end{lemma}

In particular, Observation~\ref{observation} yields with Theorem~\ref{barnettetheorem} that every $3$-connected graph contains a subdivision of the $K_4$ (Theorem of J. Isbell~\cite{Barnette1969}). Each graph $G_i$ in our construction sequence can be identified with the unique subdivision $S_i$ contained in $G$. Conversely, $G_i = \emph{smooth}(S_i)$ for all $0 \leq i \leq z$, since smoothing a graph is exactly the inverse operation of subdividing a graph without nodes of degree two. The nodes $x$ in $S_i$ with $deg(x) \geq 3$ are called \emph{real} nodes, because they correspond to nodes in $G_i$. Real nodes have at least $3$ neighbors in $G_i$, because $G_i$ is $3$-connected.

\begin{figure}[htb]
	\centering
	\subfigure[$K_4~=~~~~~$ $G_0~=~\emph{smooth}(S_0)$]{
		\includegraphics[scale=0.45]{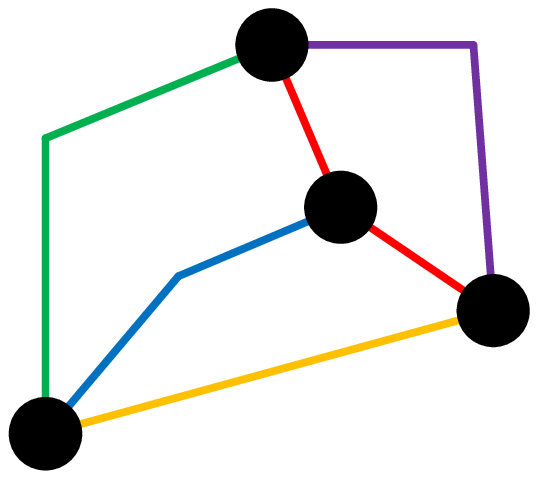}
		\label{fig:G0}
	}
	\hfill
	\subfigure[$G_1~=~\emph{smooth}(S_1)$]{
		\includegraphics[scale=0.45]{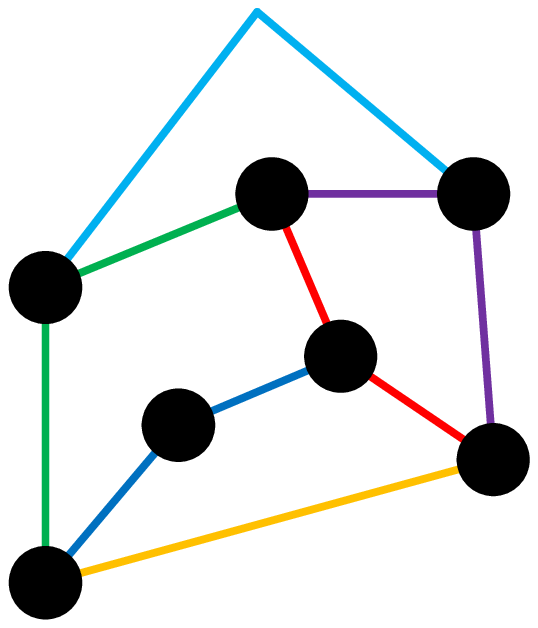}
		\label{fig:G1}
	}
	\hfill
	\subfigure[$G_2~=~\emph{smooth}(S_2)$]{
		\includegraphics[scale=0.45]{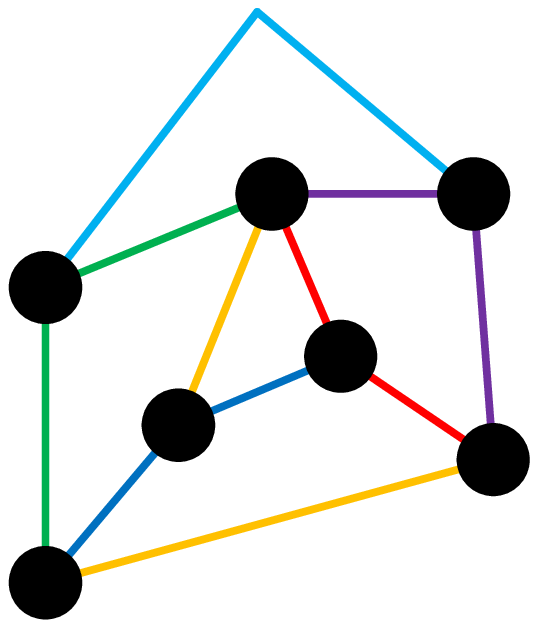}
		\label{fig:G2}
	}
	\hfill
	\subfigure[$G_3~=~G$]{
		\includegraphics[scale=0.45]{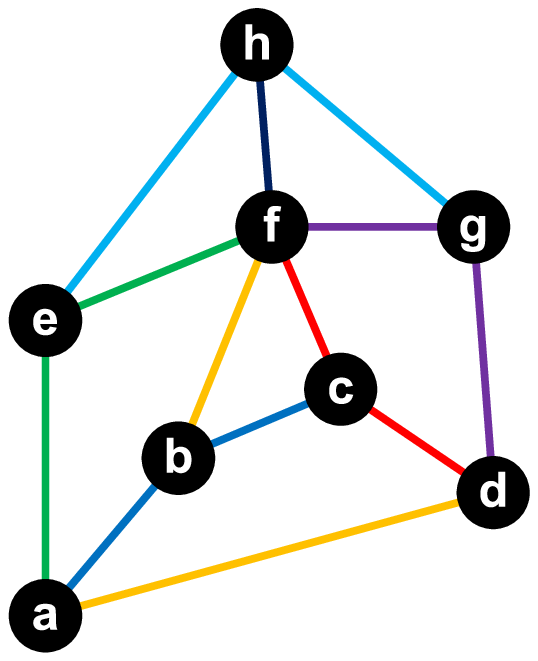}
		\label{fig:G3}
	}
	\linebreak	
	\subfigure[$S_0$]{
		\includegraphics[scale=0.45]{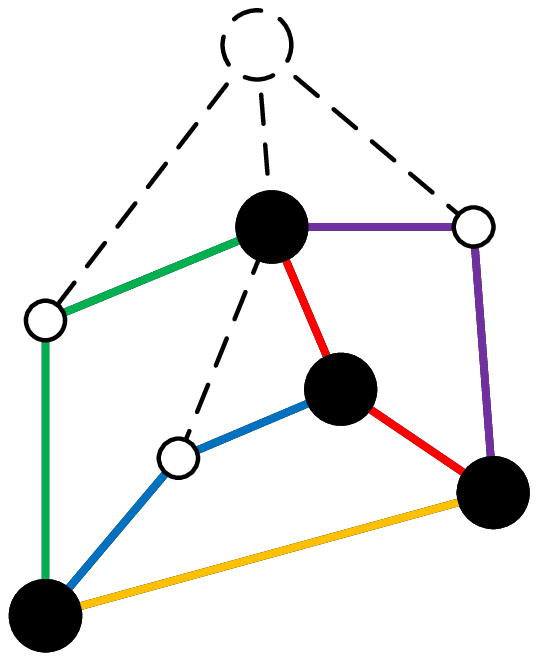}
		\label{fig:S0}
	}
	\hfill
	\subfigure[$S_1$]{
		\includegraphics[scale=0.45]{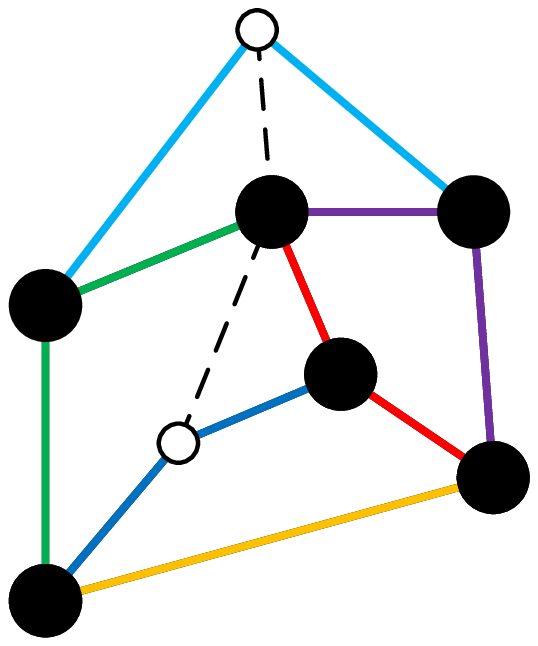}
		\label{fig:S1}
	}
	\hfill
	\subfigure[$S_2$]{
		\includegraphics[scale=0.45]{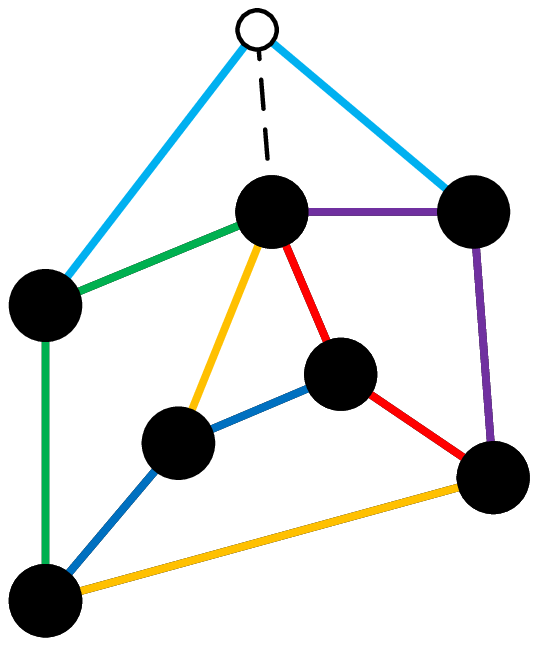}
		\label{fig:S2}
	}
	\hfill
	\subfigure[$S_3~=~G$]{
		\includegraphics[scale=0.45]{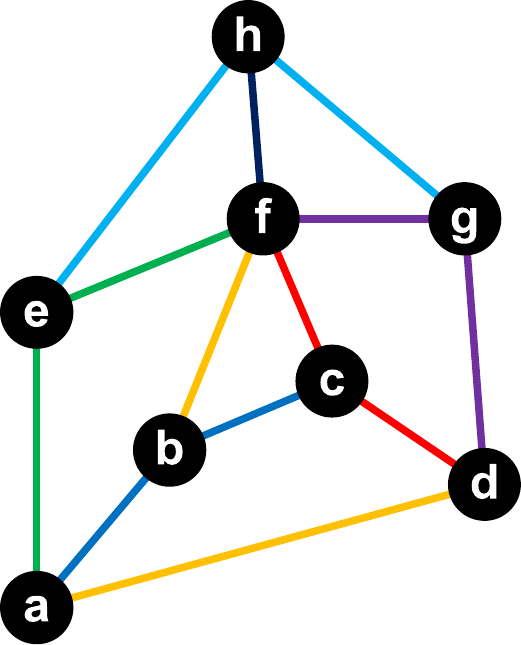}
		\label{fig:S3}
	}
	\caption{The graphs $G_0,\ldots,G_z$ and $S_0,\ldots,S_z$ of a construction sequence of $G$. On graphs $S_i$, the dashed edges and nodes are in $G$ but not in $S_i$ and nodes depicted in black are \emph{real} nodes. For example, the path $C_0 = e \rightarrow h \rightarrow g$ is a \emph{\BG-path} for $S_0$, yielding $S_1$. The \emph{links} of $S_1$ are the paths $C_0$, $a \rightarrow b \rightarrow c$ and the single edges $ae$, $ef$, $fc$, $cd$, $da$, $fg$, $gd$.}
	\label{fig:graphsequences}
\end{figure}

Note that in non-basic construction sequences $\emph{smooth}(S_i)$ can have parallel edges, although $S_i$ is always simple. We define the \emph{links} of each $S_i$ to be the unique paths in $S_i$ with only their endnodes being real. The links of $S_i$ partition $E(S_i)$ because $S_i$ is $2$-connected, has therefore minimum degree two and is not a cycle. Let two links be \emph{parallel} if they share the same endnodes.

\bigskip\bigskip\bigskip

\begin{definition}\label{bgpathdefinition}
A \emph{\BG-path for $S_i$} is a path $P = x \rightarrow y$ in $G$ with the following properties:
\begin{enumerate}
	\item $S_i \cap P = \{x,y\}$\label{bgpathdefinition1}
	\item $x$ and $y$ are not both contained in a link of $S_i$ except as endnodes\label{bgpathdefinition2}
	\item $x$ and $y$ are not inner nodes of links of $S_i$ that are parallel\label{bgpathdefinition3}
\end{enumerate}
\end{definition}

It is easy to see that every \BG-path for $S_i$ corresponds to a \BG-operation on $G_i$ and vice versa. We will exploit this duality in the next section.

In general, construction sequences are not bound to start with the $K_4$. Titov and Kelmans~\cite{Titov1975,Kelmans1978} extended Theorem~\ref{barnettetheorem} by proving the existence of a construction sequence even when starting with arbitrary $3$-connected graphs $G_0$ instead of the $K_4$, as long as a subdivision of $G_0$ is contained in $G$. This is a generalization, since every $3$-connected graph contains a subdivision of the $K_4$ by Observation~\ref{observation}.

\begin{theorem}\label{kelmans}\emph{~\cite{Kelmans1978,Titov1975}}
Let $G_0$ be a $3$-connected graph. Then a simple graph $G$ is $3$-connected and contains a subdivision of $G_0$ if and only if $G$ can be constructed from $G_0$ using basic \BG-operations.
\end{theorem}




%% file: Existence.tex
Both Theorems~\ref{barnettetheorem} and~\ref{kelmans} choose a very special subdivision of the $K_4$ (resp. $G_0$) on which the construction sequence starts, in fact one having the maximum number of edges in $G$. The construction sequence is then obtained by adding longest \BG-paths. Unfortunately, computing these depends heavily on solving the longest paths problem, which is known to be NP-hard even for $3$-connected graphs~\cite{Garey1976}.

This gives rise to the question whether Theorems~\ref{barnettetheorem} and~\ref{kelmans} can be strengthened to start at a \emph{prescribed} subdivision $H \subseteq G$ of $G_0$ instead of an arbitrary one. Note that this is equivalent to the constraint $S_0 = H$. Such a result would provide an efficient computational approach to construction sequences, since it allows us to search the neighborhood of $H$ for \BG-paths, yielding a new prescribed subdivision of a $3$-connected graph.

\begin{wrapfigure}[10]{R}{5cm}
	\centering
	\includegraphics[scale=0.4]{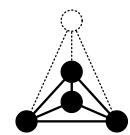}
	\caption{Every possible \BG-operation adds a parallel edge.}
	\label{fig:counterexample}
\end{wrapfigure}

However, when restricted to basic operations it is not possible to prescribe $H$, as the minimal counterexample in Figure~\ref{fig:counterexample} shows: Consider the graph $G$ consisting of a $K_4 = H$ depicted in black with an additional node connected to three nodes of the $K_4$. Then every \BG-path for $H$ will create a parallel link, although $G$ is simple. But what if we drop the condition that construction sequences have to be basic? The following theorem shows that at this expense we can indeed start a construction sequence from any prescribed subdivision.

\begin{theorem}\label{multipleconstruction}
Let $G$ be a $3$-connected graph and $H \subset G$ with $H$ being a subdivision of a $3$-connected graph. Then there is a \BG-path for $H$ in $G$. Moreover, every link of $H$ of length at least $2$ contains an inner node on which a \BG-path for $H$ starts.
\end{theorem}
\begin{proof}
We distinguish two cases.
\begin{itemize}
	\item $H \neq \emph{smooth}(H).$\\
	Then links of length at least $2$ exist in $H$ and we pick an arbitrary one of them, say $T$. Let $x$ be an inner node of $T$, and let $Q$ be the set of paths in $G$ from $x$ to a node in $V(H) \setminus V(T)$ avoiding the endnodes of $T$ (see Figure~\ref{fig:dfs}). By the $3$-connectedness of $G$, the set $Q$ cannot be empty and every path in $Q$ fulfills Definition~\ref{bgpathdefinition}.\ref{bgpathdefinition2}. There is at least one path $P = x \rightarrow y$ in $Q$ with $y$ being not contained in a parallel link of $T$, because otherwise the endnodes of $T$ would form a separation pair. Let $x'$ be the last node in $P$ that is in $T$ or in a parallel link of $T$ and let $y'$ be the first node after $x'$ that is in $V(H)$. Then $x' \rightarrow y'$ has properties~\ref{bgpathdefinition}.\ref{bgpathdefinition1} and~\ref{bgpathdefinition}.\ref{bgpathdefinition3} and is a \BG-path for $H$.
	\item $H = \emph{smooth}(H).$\\
	Then $H$ consists only of real nodes and since $H \neq G$, there is a node in $V(G) \setminus V(H)$ or an edge in $E(G) \setminus E(H)$. At first, assume that there is a node $x \in V(G) \setminus V(H)$. Then, by the $2$-connectedness of $G$ and Fan Lemma~\ref{fanlemma} we can find a path $P = y_1 \rightarrow x \rightarrow y_2$ with no other nodes in $H$ than $y_1$ and $y_2$. For $P$ the properties~\ref{bgpathdefinition}.\ref{bgpathdefinition1}-\ref{bgpathdefinition}.\ref{bgpathdefinition3} hold, because no link in $H$ can have inner nodes. Let now $V(G) = V(H)$ and $e$ an edge in $E(G) \setminus E(H)$. Then $e$ must be a \BG-path for $H$, since both endnodes are real.
\end{itemize}
\end{proof}

In Theorem~\ref{multipleconstruction}, non-basic operations can only occur in the case $H = \emph{smooth}(H)$ when a path through a node of $V(G) \setminus V(H)$ is chosen. Although we cannot avoid that, it is possible to obtain a basic construction by augmenting the \BG-operations with a fourth operation~\ref{operation4}.

\begin{enumerate}[label=(\alph*), start=4]
	\item connect a new node to three distinct nodes
	\label{operation4}
\end{enumerate}

Operation~\ref{operation4} preserves $3$-connectedness with Lemma~\ref{expansionlemma} and is basic, because each new edge ends on the new node. Whenever we encounter a node in $V(G) \setminus V(H)$ in Theorem~\ref{multipleconstruction}, we know by the Fan Lemma~\ref{fanlemma} and the $3$-connectedness of $G$ that there are three internally node-disjoint paths to real nodes in $H$ with all inner nodes being in $V(G) \setminus V(H)$. Adding these paths to $H$ is called an \emph{expand} operation and corresponds to operation~\ref{operation4} in the smoothed graph. This gives the following result.

\begin{theorem}\label{constructions}
Let $G$ be a simple graph and let $H$ be a subdivision of a $3$-connected graph. Then
\begin{align}
	&\ G \text{ is $3$-connected and $H \subseteq G$}\notag\\
	\Leftrightarrow &\ \delta(G) \geq 3 \text{ and } \exists \text{ construction sequence from $H$ to $G$ using \BG-paths}\label{theoremitem1}\\
	\Leftrightarrow &\ \delta(G) \geq 3 \text{ and } \exists \text{ basic construction sequence from $H$ to $G$ using \BG-paths}\notag\\
	&\ \text{and the expand operation}\label{theoremitem2}
\end{align}
\end{theorem}
\begin{proof}
Let $G$ be $3$-connected and $H \subseteq G$. Then $\delta(G) \geq 3$ holds and if $H = G$, the desired construction sequences are empty and exist. If $H \subset G$, we can apply Theorem~\ref{multipleconstruction} iteratively with or without the additional expand operation and the construction sequences exist as well. For the sufficiency part, both construction sequences imply $H \subseteq G$, since only paths are added to construct $G$. Additionally, $G$ must be $3$-connected, as adding \BG-paths to each $S_i$ preserves $S_{i+1}$ to be a subdivision of a $3$-connected graph with Theorem~\ref{barnettetheorem}, and $\delta(G) \geq 3$ ensures that the last subdivision $G$ of a $3$-connected graph is $3$-connected itself.
\end{proof}

%% file: Computational.tex
\begin{figure}
	\centering
	\includegraphics[scale=0.4]{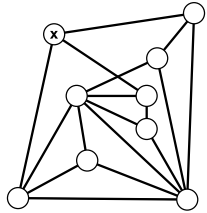}
	\caption{A $3$-connected graph having a node $x$ of degree $3$ with no incident edge being removable.}
	\label{fig:Degree3NodeWithoutRemovableEdge}
\end{figure}

A straight-forward algorithm to compute Barnette and Gr\"unbaum's construction sequence of a $3$-connected graph is to search iteratively for removable edges. But in contrast to the algorithm in Section~\ref{tuttecharacterization} that computes contractible edges, this approach only leads to an $O(n^3)$ algorithm. The reason for the additional factor of $n$ is that not all nodes with degree $3$ must have an incident removable edge (see Figure~\ref{fig:Degree3NodeWithoutRemovableEdge} for a counterexample on $9$ nodes) and we have to try every edge in the worst case. Computing \BG-paths instead of \BG-operations allows us to obtain better running times, but first we need to know how exactly construction sequences can be represented.

An obvious representation of a construction sequence $Q$ would be to store the graph $G_0 = \emph{smooth}(H)$ and in addition every \BG-operation, which gives the sequence $G_0,\ldots,G_z = G$. Unfortunately, the graphs $G_i$ are not necessarily subgraphs of $G_{i+1}$, so we have to take care of relabeled edges when specifying each operation.

Whenever an edge $e$ is subdivided as part of an operation~\ref{operation2} or~\ref{operation3}, we specify it by its index in $G_i$ followed by assigning new indices for the new degree-two node and one of the two new separated edge parts in $G_{i+1}$. The other edge part keeps the index of $e$.

Similarly, on operations~\ref{operation1} and~\ref{operation2}, real endnodes of the added edge are specified by their indices in $G_i$. We assign a new index for the added edge in $G_{i+1}$, too. Finally, we have to impose the constraint that $G_z$ is not just isomorphic but identical to $G$, meaning that nodes and edges of $G_z$ and $G$ are labeled by exactly the same indices, since otherwise we would have to solve the graph isomorphism problem to check that $Q$ really constructs~$G$.

\begin{figure}
	\centering
	\includegraphics[scale=0.7]{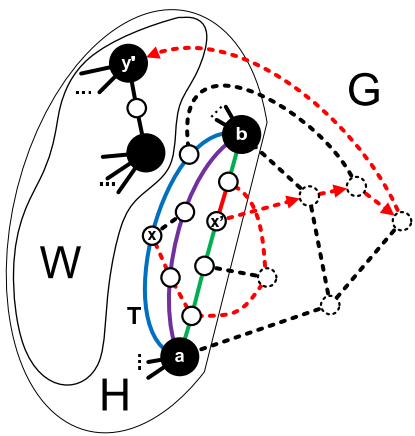}
	\caption{The case $H \neq \emph{smooth}(H)$. Dashed edges are in $E(G) \setminus E(H)$, arrows depict the \BG-path $x' \rightarrow y'$.}
	\label{fig:dfs}
\end{figure}

On the other hand, the identification of $G_i$ with a subgraph in $G$ allows us to represent $Q$ without indexing issues: We just store $S_0 \subset G$ and the \BG-paths $C_0,\ldots,C_{z-1}$. Hence, we can represent each construction sequence $Q$ of $G$ in the following two ways.

\begin{itemize}
	\item \emph{Edge representation}: Represent $Q$ by $G_0$ and a sequence of \BG-operations, along with specifying new and old indices for each operation, such that $G_z$ and $G$ are labeled the same.
	\item \emph{Path representation}: Represent $Q$ by $S_0$ and \BG-paths $C_0,\ldots,C_{z-1}$.
\end{itemize}

Both representations refer to the same sequence of graphs $G_0,\ldots,G_z$ and are of size $\theta(m)$, assuming the uniform cost model. The next lemma states that it does not matter which of the two representations we compute.

\begin{lemma}\label{representation}
The edge and path representations of a construction sequence $Q$ can be transformed into each other in $O(m)$ time. Moreover, the representation computed is a unique representation of $Q$.
\end{lemma}
\begin{proof}
Omitted.
\end{proof}

%% file: Certifying.tex
We use construction sequences in the path representation as a certificate for the $3$-connectedness of graphs. This leads to a new, certifying method for testing graphs on being $3$-connected. The total running time of this method is $O(n^2)$, however this is dominated by the time needed for finding the construction sequence and every improvement made there will automatically result in a faster $3$-connectedness test. The input graph is a multigraph and does not have to be biconnected nor connected. We follow the steps:

\begin{itemize}
	\item Apply preprocessing of Nagamochi and Ibaraki to the graph and get $G$ in $O(n+m)$\\(This improves the total running time by decreasing the number of edges to $O(n)$.)
	\item Try to compute a $K_4$-subdivision $S_0$ in $G$ and prescribe it in $O(n)$
	\begin{itemize}
		\item Failure: Return a separation pair
	\end{itemize}
	\item Try to compute a construction sequence from $S_0$ to $G$ in $O(n^2)$
	\begin{itemize}
		\item Success: Return the construction sequence
		\item Failure: Return a separation pair
	\end{itemize}
\end{itemize}

\begin{wrapfigure}[17]{R}{5cm}
	\centering
	\includegraphics[scale=0.6]{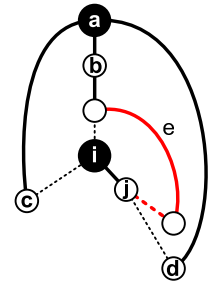}
	\caption{Finding a $K_4$-subdivision. Dashed edges can be (empty) paths, arcs depict backedges.}
	\label{fig:findk4}
\end{wrapfigure}

The preprocessing step preserves the graph to be $3$-connected or to be not $3$-connected. We first describe how to find a $K_4$-subdivision by one Depth First Search (DFS), which as a byproduct eliminates self-loops and parallel edges and sorts out graphs that are not connected or have nodes with degree at most $2$. Let $a$ (resp. $b$) be the node in the DFS-tree $T$ that is visited first (resp. second). If $G$ is $3$-connected, then $a$ and $b$ have exactly one child, otherwise they form a separation pair. We choose two arbitrary neighbors $c$ and $d$ of $a$ that are different from $b$ (see Figure~\ref{fig:findk4}). W.l.o.g., let $d$ be visited later by the DFS than $c$. Let $i \neq b$ the least common ancestor of $c$ and $d$ in $T$. As $d \neq i$ must hold, let $j$ be the child of $i$ that is contained in the path $i \rightarrow d$ in $T$.

If $G$ is $3$-connected, we can find a backedge $e$ that starts on a node $z$ in the subtree rooted at $j$ and ends on an inner node $z'$ of $a \rightarrow i$ in time $O(n)$. If $e$ does not exist, $a$ and $i$ form a separation pair, otherwise we have found a $K_4$-subdivision with real nodes $a$, $i$, $z$ and $z'$. The paths connecting this real nodes in $T$ together with the three visited backedges constitute the $6$ paths of the $K_4$-subdivision.

Once the $K_4$-subdivision $S_0$ is found, we follow the lines of Theorem~\ref{multipleconstruction} and try to construct the path representation $C_0,\ldots,C_{z-1}$. If favored, this can be transformed to an edge representation in $O(m)$ later. We assign an index for every link and store it on each of the inner nodes of that link. Moreover, we maintain pointers for each link to its endnodes.

In case $H \neq \emph{smooth}(H)$ of Theorem~\ref{multipleconstruction} we pick an arbitrary node $x$ of degree two. Let $T = a \rightarrow b$ be the link that contains $x$ and let $W$ be the set of nodes $V(H) \setminus V(T)$ minus all nodes in parallel links of $T$ (see Figure~\ref{fig:dfs}). We compute the path $P = x \rightarrow y'$ by temporarily deleting $a$ and $b$ and performing a DFS on $x$ that stops on the first node $y' \in W$. We can check whether a node lies in a parallel link of $T$ in constant time by comparing the endnodes of its containing link with $a$ and $b$. Thus, the subpath $x' \rightarrow y'$ with $x'$ being the last node contained in $T$ or in a parallel link of $T$ is a \BG-path and can be found efficiently. The links and their indices can be updated in $O(n)$.

Similarly, in case $H = \emph{smooth}(H)$ we delete temporarily all edges in $E(H)$ and start a DFS on a node $x \in V(H)$ that has an incident edge in the remaining graph. The traversal is stopped on the first node $y \in V(H) \setminus \{x\}$. The path $x \rightarrow y$ is then the desired \BG-path and we conclude that for $3$-connected graphs the construction sequence can be found in time $O(n^2)$.

Otherwise, $G$ is not $3$-connected and no construction sequence can exist with Theorem~\ref{constructions}. In that case a DFS starting at node $x$ fails to find a new \BG-path for some subdivision $H \subset G$. If $H \neq \emph{smooth}(H)$, the endnodes of the link that contains $x$ must form a separation pair. Otherwise, $H = \emph{smooth}(H)$ and $x$ must be a cut vertex. Thus, if $G$ is not $3$-connected, the algorithm returns always a separation pair or cut vertex.

If $G$ is simple, the construction sequence can be transformed to the basic construction sequence~\eqref{theoremitem2} with the following Lemma.

\begin{lemma}\label{transform}
For simple graphs $G$, the construction sequences~\eqref{theoremitem1} and~\eqref{theoremitem2} can be transformed into each other in $O(m)$.
\end{lemma}
\begin{proof}
Omitted.
\end{proof}

\begin{wrapfigure}[10]{R}{5cm}
	\centering
	\includegraphics[scale=0.6]{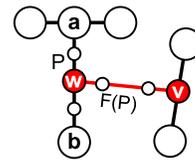}
	\caption{No expand operation can be formed.}
	\label{fig:notexpand}
\end{wrapfigure}

\begin{theorem}\label{algorithm}
The construction sequences~\eqref{theoremitem1} and~\eqref{theoremitem2} can be computed in $O(n^2)$ and establish a certifying $3$-connectedness test with the same running time.
\end{theorem}

\subsection{Verifying the Construction Sequence}
It is essential for a certificate that it can be easily validated. We could do this by transforming the path representation to the edge representation using Lemma~\ref{representation} and checking the validity of the \BG-operations by comparing indices, but there is a more direct way. First, it can be checked in linear time that all \BG-paths $C_0,\ldots,C_{z-1}$ are paths in $G$ and that these paths partition $E(G) \setminus E(S_0)$. We try to remove the \BG-paths $C_{z-1},\ldots,C_0$ from $G$ in that order (i.\,e., we delete the paths followed by smoothing its endnodes). If the certificate is valid, this is well defined as all removed \BG-paths are then edges. On the other hand we can detect longer \BG-paths $|C_i| \geq 2$ before their removal, in which case the certificate is not valid, since then the inner nodes of $C_i$ are not attached to \BG-paths $C_j$, $j > i$.

\begin{figure}
	\centering
	\subfigure[Either $a$ or $b$ has degree $2$.]{
		\includegraphics[scale=0.6]{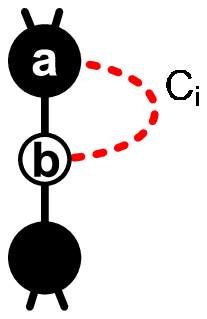}
		\label{fig:check1}
	}
	\hspace{1cm}
		\subfigure[Both, $a$ and $b$, have degree $2$.]{
		\includegraphics[scale=0.6]{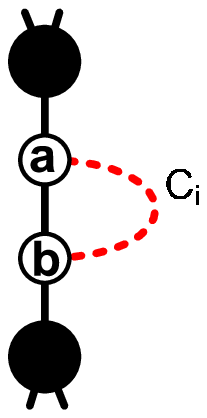}
		\label{fig:check2}
	}
	\caption{Cases where~\ref{bgpathdefinition}.\ref{bgpathdefinition2} fails when $a \in N(b)$.}
	\label{fig:check}
\end{figure}

We verify that every removed $C_i = ab$ corresponds to a \BG-operation by using Definition~\ref{bgpathdefinition} of \BG-paths, and start with checking that $a$ and $b$ lie in our current subgraph for condition~\ref{bgpathdefinition}.\ref{bgpathdefinition1}.

Conditions~\ref{bgpathdefinition}.\ref{bgpathdefinition2} and~\ref{bgpathdefinition}.\ref{bgpathdefinition3} can now be checked in constant time: Consider the situation immediately after the deletion of $ab$, but before smoothing $a$ and $b$. Then all links in our subgraph are single edges, except possibly the ones containing $a$ and $b$ as inner nodes.

Therefore,~\ref{bgpathdefinition}.\ref{bgpathdefinition2} is not met for $C_i$ if $a$ is a neighbor of $b$ and at least one of the nodes $a$ and $b$ has degree two (see Figures~\ref{fig:check} for possible configurations). Condition~\ref{bgpathdefinition}.\ref{bgpathdefinition3} is not met if $N(a) = N(b)$ and both $a$ and $b$ have degree two. Both conditions can be easily checked in constant time. Note that encountering proper \BG-paths $C_{z-1},\ldots,C_i$ does not necessarily imply that the current subgraph is $3$-connected, since false \BG-paths $C_j$, $j < i$, can exist.

It remains to validate that the graph after removing all \BG-paths is the $K_4$. This can done in constant time by checking it on being simple and having exactly $4$ nodes of degree three.

\begin{theorem}\label{verifying}
The construction sequences~\eqref{removals}-\eqref{BGBasic} and \eqref{theoremitem1}-\eqref{theoremitem2} can be checked on validity in time linearly dependent on their length.
\end{theorem}

%% file: Schmidt.bbl
\begin{thebibliography}{10}

\bibitem{Albroscheit2006}
S.~Albroscheit.
\newblock Ein {A}lgorithmus zur {K}onstruktion gegebener
  3-zusammenh\"{a}ngender {G}raphen.
\newblock Diploma thesis, FU Berlin, 2006.

\bibitem{Barnette1969}
D.~W. Barnette and B.~Gr\"unbaum.
\newblock On {S}teinitz's theorem concerning convex 3-polytopes and on some
  properties of 3-connected graphs.
\newblock {\em Many Facets of Graph Theory, Lecture Notes in Mathematics},
  110:27--40, 1969.

\bibitem{Blum1989}
M.~Blum and S.~Kannan.
\newblock Designing programs that check their work.
\newblock In {\em STOC '89}, pages 86--97, New York, 1989.

\bibitem{Garey1976}
M.~R. Garey, D.~S. Johnson, and R.~E. Tarjan.
\newblock The planar hamiltonian circuit problem is {NP}-complete.
\newblock {\em Siam J. Comp.}, 5(4):704--714, 1976.

\bibitem{Halin1969a}
R.~Halin.
\newblock Zur {T}heorie der n-fach zusammenh\"angenden {G}raphen.
\newblock {\em Abhandlungen aus dem Mathematischen Seminar der Universit\"at
  Hamburg}, 33(3):133--164, 1969.

\bibitem{Hopcroft1973}
J.~E. Hopcroft and R.~E. Tarjan.
\newblock Dividing a graph into triconnected components.
\newblock {\em SIAM J. Comput.}, 2(3):135--158, 1973.

\bibitem{Kelmans1978}
A.~K. Kelmans.
\newblock Graph expansion and reduction.
\newblock {\em Algebraic methods in graph theory, Szeged, Hungary}, 1:317--343,
  1978.

\bibitem{Menger1927}
K.~Menger.
\newblock Zur allgemeinen {K}urventheorie.
\newblock {\em Fund. Math.}, 10:96--115, 1927.

\bibitem{Nagamochi1992}
H.~Nagamochi and T.~Ibaraki.
\newblock A linear-time algorithm for finding a sparse k-connected spanning
  subgraph of a k-connected graph.
\newblock {\em Algorithmica}, 7(1-6):583--596, 1992.

\bibitem{Thomassen1981}
C.~Thomassen.
\newblock Kuratowski's theorem.
\newblock {\em Journal of Graph Theory}, 5(3):225--241, 1981.

\bibitem{Thomassen2006}
C.~Thomassen.
\newblock Reflections on graph theory.
\newblock {\em Journal of Graph Theory}, 10(3):309--324, 2006.

\bibitem{Titov1975}
V.~K. Titov.
\newblock {\em A constructive description of some classes of graphs}.
\newblock PhD thesis, Moscow, 1975.

\bibitem{Tutte1961}
W.~T. Tutte.
\newblock A theory of 3-connected graphs.
\newblock {\em Indag. Math.}, 23:441--455, 1961.

\bibitem{Tutte1966}
W.~T. Tutte.
\newblock Connectivity in graphs.
\newblock In {\em Mathematical Expositions}, volume~15. University of Toronto
  Press, 1966.

\bibitem{Vo1983}
K.-P. Vo.
\newblock Finding triconnected components of graphs.
\newblock {\em Linear and Multilinear Algebra}, 13:143--165, 1983.

\bibitem{Vo1983a}
K.-P. Vo.
\newblock Segment graphs, depth-first cycle bases, 3-connectivity, and
  planarity of graphs.
\newblock {\em Linear and Multilinear Algebra}, 13:119--141, 1983.

\bibitem{West2001}
D.~B. West.
\newblock {\em Introduction to Graph Theory}.
\newblock Prentice Hall, 2001.

\end{thebibliography}
